\documentclass[10pt,draftcls,peerreview,onecolumn]{IEEEtran}
\usepackage{setspace}

\usepackage{amsmath,amsfonts,amssymb}
\usepackage{enumerate}
\usepackage{color}
\usepackage{xcolor}
\usepackage{mathrsfs}
\usepackage{algpseudocode}
\usepackage{cite}
\usepackage{soul}
\usepackage{subfigure}
\usepackage{epsf}
\usepackage{bm}
\usepackage{latexsym}
\usepackage[pdftex]{graphicx}
\usepackage{epstopdf} 
\DeclareGraphicsExtensions{.jpg,pdf,.png}
\usepackage{amsthm}
\usepackage[]{algorithm2e}
\newtheorem{theorem}{Theorem}

\newtheorem{corollary}{Corollary}[theorem]

\theoremstyle{definition}

\theoremstyle{remark}
\newtheorem{remark}{Remark}

\ifCLASSINFOpdf
\else
\fi

\begin{document}
\doublespacing

\title{Sparsity Promoting Reconstruction of Delta Modulated Voice Samples by Sequential Adaptive Thresholds}

\author{$^\ast$ 
	Mahdi Boloursaz Mashhadi, Saber Malekmohammadi, Student Members, IEEE, and Farokh Marvasti, Senior Member, IEEE

\thanks{$^\ast$ Mahdi Boloursaz Mashhadi and Farokh marvasti are with the Advanced Communications Research Institute (ACRI), EE Department, Sharif University of Technology (SUT), Tehran, Iran. Saber Malekmohammadi is with the EE Department, University of Waterloo, ON, Canada. email: boloursaz@ee.sharif.edu.}
\thanks{}}

\maketitle


\begin{abstract}
 In this paper, we propose the family of Iterative Methods with Adaptive Thresholding (IMAT) for sparsity promoting reconstruction of Delta Modulated (DM) voice signals. We suggest a novel missing sampling approach to delta modulation that facilitates sparsity promoting reconstruction of the original signal from a subset of DM samples with less quantization noise. Utilizing our proposed missing sampling approach to delta modulation, we provide an analytical discussion on the convergence of IMAT for DM coding technique. We also modify the basic IMAT algorithm and propose the Iterative Method with Adaptive Thresholding for Delta Modulation (IMATDM) algorithm for improved reconstruction performance for DM coded signals. Experimental results show that in terms of the reconstruction SNR, this novel method outperforms the conventional DM reconstruction techniques based on lowpass filtering. It is observed that by migrating from the conventional low pass reconstruction technique to the sparsity promoting reconstruction technique of IMATDM, the reconstruction performance is improved by an average of 7.6 dBs. This is due to the fact that the proposed IMATDM makes simultaneous use of both the sparse signal assumption and the quantization noise suppression effects by smoothing. The proposed IMATDM algorithm also outperforms some other sparsity promoting reconstruction methods.
\end{abstract}

\section{introduction}
Uniform sampling is a prior technique for Analog to Digital (A/D) conversion which samples the signal amplitudes at uniformly spaced time instances. If the Shannon-Nyquist criteria is met by the sampling process, then perfect reconstruction of the original analog signal from noiseless uniform samples is guaranteed\cite{shannon, jerri}. Alternative nonuniform sampling approaches to A/D conversion such as Level Crossing Sampling \cite{MYLC, MYPATENT}, Delta Modulation (DM) \cite{nusampling, Drmarvasti2001nonuniform} and Sigma-Delta Modulation techniques \cite{Drmarvasti2001nonuniform} are also commonly used. Delta modulation \cite{DMMOD} and its improved version Adaptive Delta Modulation (ADM) \cite{ADMIET} are major waveform coding techniques that have found widespread applications in current voice and video coding systems\cite{Zhu_paper, ADM, app, Reza, Soheil}.\\

Delta modulation is characterized by the fact that each sample is coded by a single binary symbol which leads to a rather simple hardware implementation \cite{DMMOD}. This technique adds modulation noise \cite{noise} to the original signal. This noise may either be due to the slope overloading phenomenon during rapid jumps/falls of the signal or due to the oscillations of the modulator output during slowly varying portions of the signal which is called the granular noise \cite{qnoise}, \cite{grannoise}.\\

Efficient reconstruction of the original signal from DM samples is a major issue. The conventional reconstruction method is based on low pass filtering the modulator output at the demodulator \cite{Drmarvasti2001nonuniform},\cite{lprecons}. However, in this research, we migrate to the emergent theory of sparse signal processing \cite{Dramini} for reconstructing DM coded signals. To this end, we propose the missing sampling approach to Delta Modulation in which the more accurate DM samples taken from the signal during low slope portions are used to recover the original signal utilizing a specific sparse reconstruction technique. We show by simulations that by replacing the conventional low pass signal assumption with sparse signal assumption, we achieve an average improvement of 7.6 dBs in reconstructing voice signals from DM coded samples. This is due to the fact that the sparse reconstruction approach is capable of reconstructing the less significant but still considerable high frequency components that are commonly discarded by the low pass reconstruction techniques. To achieve more comprehensive results, we consider both conventional and adaptive delta modulation sampling schemes. We compare the performance of different sparse reconstruction algorithms (OMP \cite{Omp}, LASSO \cite{lasso} and IMAT \cite{IMAT}) in the proposed scenario and observe that the Iterative Method with Adaptive Tresholding (IMAT) algorithm achieves superior results. Further we modify the basic IMAT and propose the IMATDM algorithm for improved reconstruction performance in DM scenario. Finally, we analytically prove convergence of the IMAT algorithm for DM reconstruction scenario and derive the final reconstruction SNR in presence of the quantization noise.\\

The rest of paper is organized as follows. Section II explains the proposed missing sampling approach to delta modulation. In section III we explain the basic IMAT algorithm and propose its modified version IMATDM for improved reconstruction performance in the DM scenario. In section IV we analytically prove the convergence of IMAT in the proposed DM scenario. In section V, we discuss the simulation results to support our previous claims. Finally, we conclude this work in section VI.

\section{the missing sampling approach to delta modulation}
Delta Modulation (DM) is a basic analog-to-digital conversion technique used for efficient coding of voice signals. DM is the simplest form of differential pulse code modulation (DPCM) where the difference between successive samples are encoded to bit streams. In delta modulation, the transmitted samples are reduced to 1-bit symbols. Figure \ref{DMblock} gives the block diagram for the basic delta modulation/demodulation technique.

\begin{figure}[!ht]
	\centering
	\includegraphics[width=0.5\columnwidth,,height=0.375\linewidth]{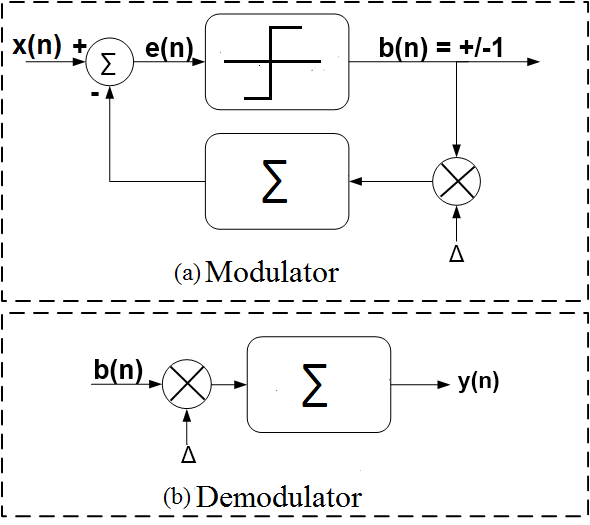}
	\caption{Block Diagram of the basic Delta Modulator/Demodulator}
	\label{DMblock}
\end{figure}

Fig \ref{fig:modulation}.(a) depicts a typical voice frame $x$ and its corresponding DM output $y$. As observed in this figure, during low activity periods where the input signal varies slowly with time, the DM output signal oscillates around the input. In contrary, in regions with slopes of higher absolute values where the input signal experiences rapid jumps/falls, DM output promotes or postpones the input. In both cases a coding error is introduced to the original signal. However, in the first case, DM samples provide a better approximation of the input signal with a maximum error of $\frac{\Delta}{2}$. 

\begin{figure}[!ht]
	\centering
	
	\subfigure[Delta Modulation ]
	{
		\includegraphics[width=0.45 \columnwidth ]{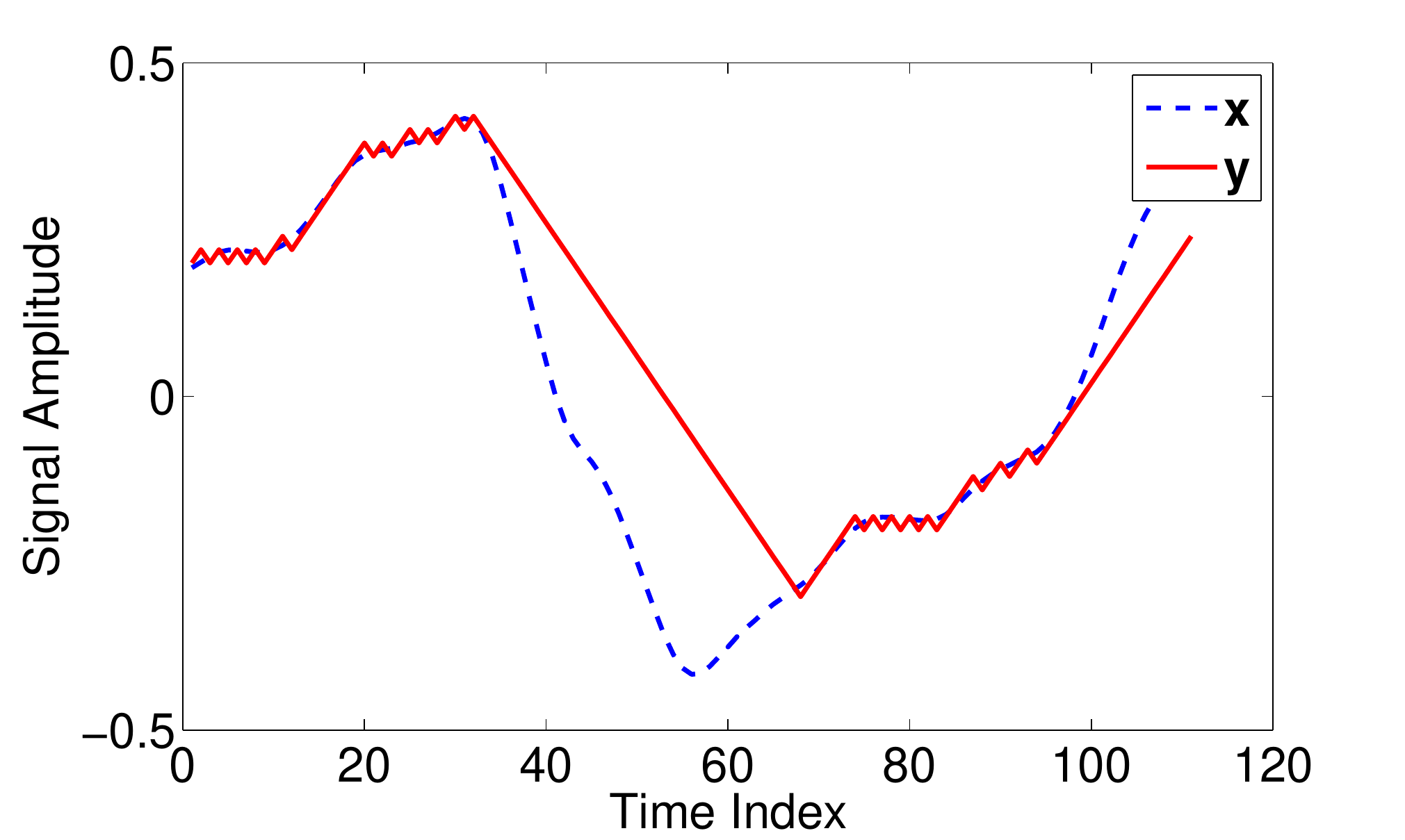}
		\label{fig:FDMsavingCW1}
	}
	\centering	
	\subfigure[The Missing Sampled DM Signal]
	{
	\includegraphics[width=0.45 \columnwidth]{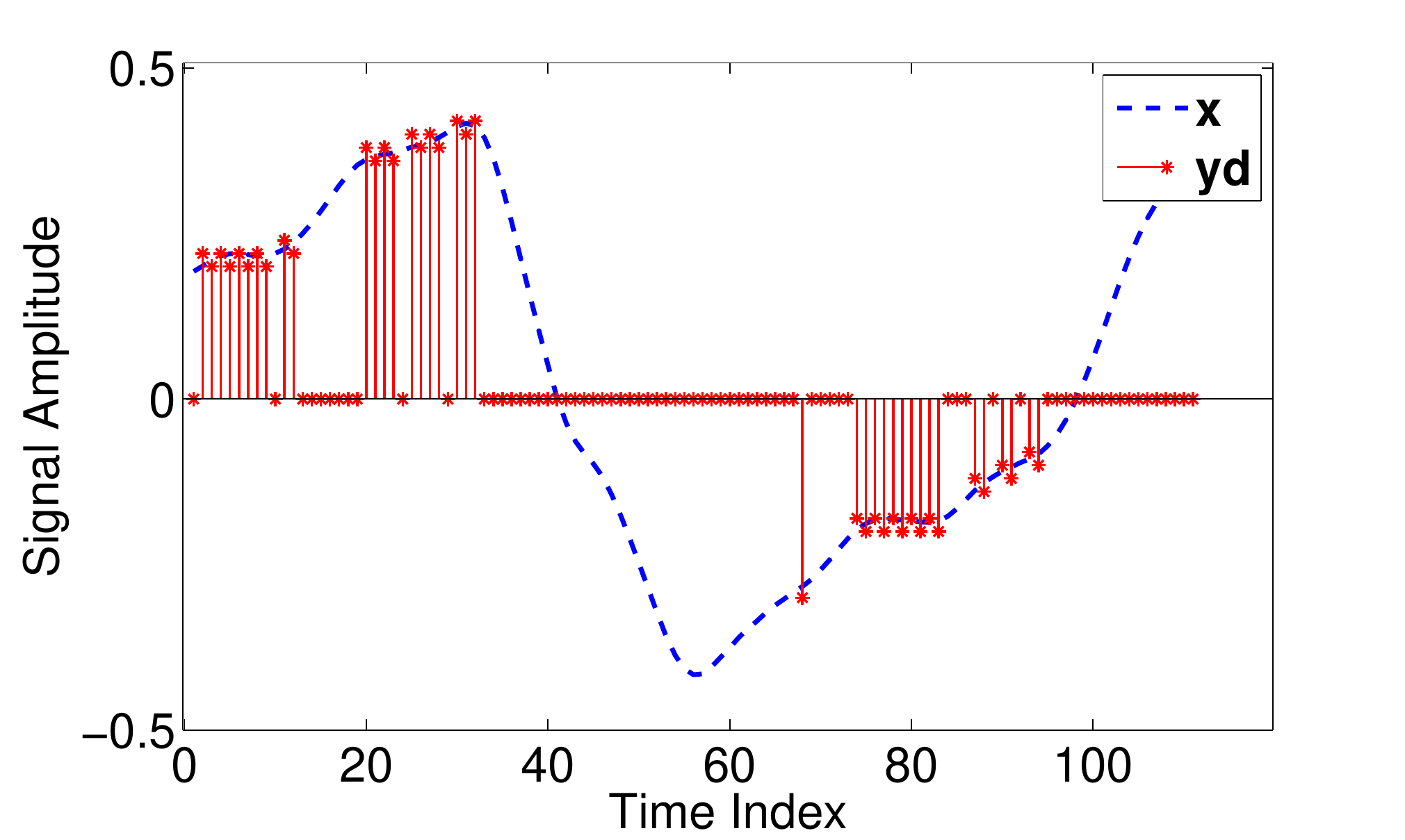}
	\label{fig:FDMsavingCW2}
	}

	\caption{The Missing Sampling Approach to Delta Modulation}
	\label{fig:modulation}
\end{figure}
 
Since emergence, DM has been accompanied by subsequent low pass filters for reconstruction. Low pass filtering enforces the presumption that the underlying signal is either inherently low pass or is low pass filtered by the analog anti-aliasing filter prior to A/D conversion. The prior anti-aliasing filter omits the less significant but still considerable high frequency components from the input signal and hence degrades the reconstruction performance of A/D conversion. To avoid this phenomenon, we utilize techniques from the emergent theory of sparse reconstruction. To utilize the sparse reconstruction techniques, we need to provide a random sub-sample of the original signal for reconstruction. In this section, we propose a novel missing sampling approach to delta modulation to achieve this goal.

In the proposed approach, the more accurate DM samples taken in low slope regions of the input signal are utilized for its reconstruction during high activity regions. In other words, the less accurate samples taken during jumps/falls are discarded (missing samples) but estimated by the proposed reconstruction algorithm from the more accurate samples. Hence, in the proposed missing sampling model, only the oscillating samples are used for signal reconstruction. These oscillations are derived from the sign alternations in the delta modulator output sequence $b(n)$. In fact, the sampling mask $d(n)$ is extracted from $b(n)$ according to (\ref{eq:mask})

\begin{align} \label{eq:mask}
d(n)=\left\{
\begin{array}{lll}
1 & if & b(n).b(n+1)=-1\\
0 & & otherwise
\end{array}
\right.
\end{align}

Fig \ref{fig:modulation}. b depicts the missing sampled signal $y_d(n)=d(n).y(n)$ corresponding to the typical voice frame in Fig \ref{fig:modulation}. For further simplification of the analytical studies, we assume that $y_d(n)$'s are independent identically distributed (iid) random variables that take one of the three values $0$, $x(n)+\Delta/2$ and $x(n)-\Delta/2$ with probabilities $1-p$, $p/2$ and $p/2$ respectively. In this formulation, $p$ denotes the sampling rate defined as the ratio of the ones in $d(n)$ to its total length. Equivalently, we can write $y_{d}(n)=d(n)(x(n)+q(n))$ in which the coding error $q(n)$'s are assumed iid random variables of Bernoulli distribution that take the values $\Delta/2$ and $-\Delta/2$ with equal probabilities. $d(n)$'s are also iid random variables of $Bernoulli(p)$ distribution and $q(n)$'s are assumed independent of the sampling mask elements $d(n)$'s.

It is noteworthy that the $\Delta$ parameter plays a key role in our proposed model. A larger $\Delta$ value gives a higher sampling ratio (i.e. the ratio of the reliable samples utilized for reconstruction to the total number of samples) but introduces more error to the samples utilized for reconstruction.

Finally, it should be noted that, there are different variants to the basic DM coding technique  \cite{ADMIET}, \cite{ADM}, \cite{Prosalentis2}, etc. These variants mostly improve over the basic DM by proposing novel ideas to decrease the quantization noise at the modulator. Hence, although these techniques provide more accurate DM samples, but they still utilize the same low pass filtering stage for demodulation. Our proposed sparsity promoting reconstruction algorithms are applicable to any variant of the DM technique to improve the output signal quality. To confirm this claim, we also expand our proposed missing sampling model to the Adaptive Delta Modulation (ADM) scheme.  In this scheme the $\Delta$ parameter is adapted to the variations of the input signal. During low activity periods of the input, $\Delta$ is decreased exponentially to achieve more accurate samples with less coding error. Alternatively, in regions where DM output is promoted or postponed, $\Delta$ is increased exponentially in order to faster reach the input signal.

\section{sparsity promoting reconstruction by IMAT and IMATDM}
 
To the best of the authors' knowledge, this is the first work in the literature that considers DM coding in the sparse reconstruction paradigm. In fact, prior researches have considered the underlying signal to be low pass or band-limited \cite{Drmarvasti2001nonuniform},\cite{lprecons} but not spectrally sparse. In the this section, we demonstrate our sparsity promoting reconstruction technique namely, the Iterative Method with Adaptive Thresholding (IMAT) and its improved version IMAT for DM reconstruction (IMATDM) algorithms in Subsection A. Later in Subsection B,  compares the structures of the sparsity promoting reconstruction methods of IMAT and IMATDM with the traditional lowpass filtering technique for reconstruction of DM voice samples. Later in section III, it is shown by extensive simulations that the proposed sparsity promoting IMAT and IMATDM algorithms not only outperform the conventional low pass reconstruction techniques, but also outperform other sparsity promoting reconstruction techniques of OMP, and Lasso.\\

\subsection{IMAT and IMATDM algorithms}

The basic IMAT algorithm was originally proposed for sparse signal reconstruction from missing samples in \cite{IMAT}. IMAT proved to outperform Orthogonal Matching Pursuit (OMP) and Iterative Hard Thresholding (IHT) algorithms in some applications regarding performance or complexity \cite{IMAT, Ashi1}. The basic IMAT algorithm progressively extracts the sparse signal components by iterative thresholding of the estimated signal promoting sparsity. The reconstruction formula for IMAT is given by (\ref{eq:IMAT})

\begin{align}\label{eq:IMAT}
	\text{IMAT:}\quad x_{k+1}(n)=\lambda y_d(n)+ \left(1-\lambda d(n)\right)T(x_k)(n)
\end{align}

In \eqref{eq:IMAT}, $x(n)$ and $x_k(n)$ denote the original signal and its reconstructed version at the $k$'th algorithm iteration and  $\lambda$ is the relaxation parameter that controls the convergence rate of the algorithm. $T(.)$ denotes the thresholding operator and $d(n)$ is the binary missing sampling mask extracted according to (\ref{eq:mask}). 

\begin{align}\label{eq:rabete1}
	d(n) &= \sum_{i=1}^{pN}\delta(n-n_i)
\end{align}
In \eqref{eq:rabete1}, $\delta(.)$ denotes the Kronecker delta function, $N$ is the voice frame length, $0\leq n_i \leq N-1 $ denotes the missing sampling time instances and $0 \leq p \leq 1$ is the sampling rate.

The thresholding block transforms the input signal to the sparse domain (FFT, DCT, etc.), sets the signal components with absolute values below the threshold to zero and finally transforms back to the original time domain. The threshold value Th is decreased exponentially by $Th(k)=\beta e^{\alpha k}$ where $k$ is the iteration number. The algorithm performance is less dependent on the choice of the algorithm parameters $\beta$ and $\alpha$ but these parameters are optimized empirically for the best performance in simulations. The block diagram for the basic IMAT algorithm is given in fig. \ref{fig:IMATB}. DT is the Discrete Transform (the simple FFT transfrom in this research) and converts the input signal to the  sparse domain. Similarly, IDT is its inverse transform (IFFT).
\begin{figure}[h]
	\center{
	\includegraphics[width=0.45\columnwidth,,height=0.4\linewidth]{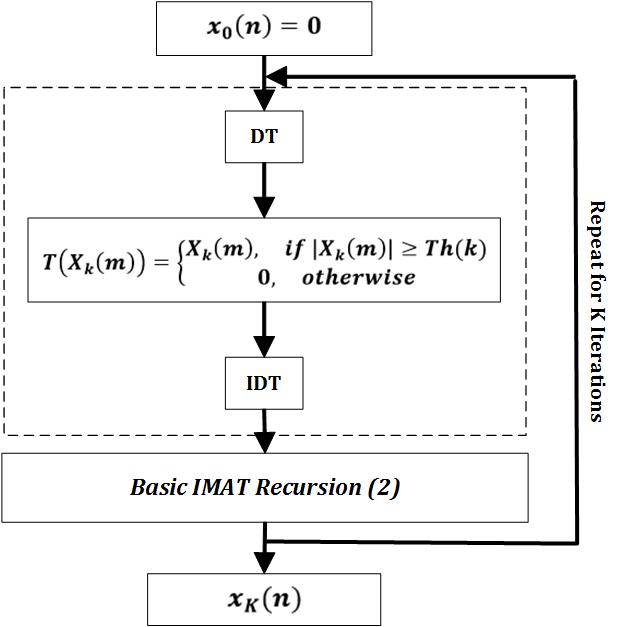}
	\caption{IMAT Block Diagram}
	\label{fig:IMATB}
}
\end{figure}

Having demonstrated the basic IMAT algorithm, we proceed with IMATDM. According to fig. \ref{fig:smooth}, IMATDM is obtained by adding a partial smoothing block prior to IMAT. This block performs a moving averaging of length $l$ on $y_d(n)$. As mentioned before, although DM samples are considered as a good approximation of the original signal during low activity intervals, they are still oscillating around the input in these regions and hence include a coding error $q(n)$. In fact, as investigated in section IV, reconstructing the signal from these noisy samples results in a noticeable performance degradation according to the power of the coding error which is proportional to $\Delta^2$. As these samples are oscillating around the original signal, it is obvious that a moving averaging of an even length $l$ causes the coding error terms to cancel out. On the other hand, as the original signal is slowly varying in these regions and does not experience sharp jumps/falls, the averaging operation does not cause much distortion. Applying this technique, more the samples used by IMAT for reconstruction are significantly more accurate and reconstruction leads to improved SNR values as observed by simulations in section V.

\subsection{Sparsity promoting vs. lowpass reconstruction}

\begin{figure}[!ht]
	\centering
	\subfigure[IMATDM Block Diagram]
	{	  
		\includegraphics[width=0.45\columnwidth]{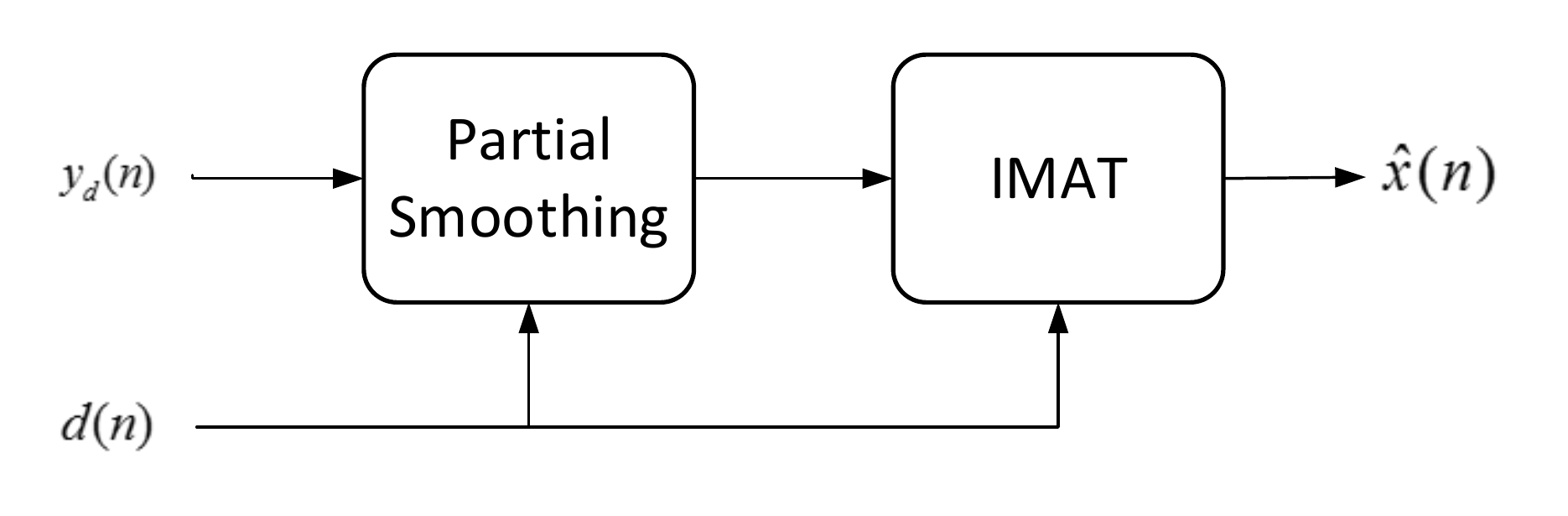}
		\label{fig:smooth}
	}
		\\
	\centering	
	\subfigure[The Conventional DM Reconstruction Method by Low pass Filtering]
	{
		\includegraphics[width= 0.4\columnwidth ]{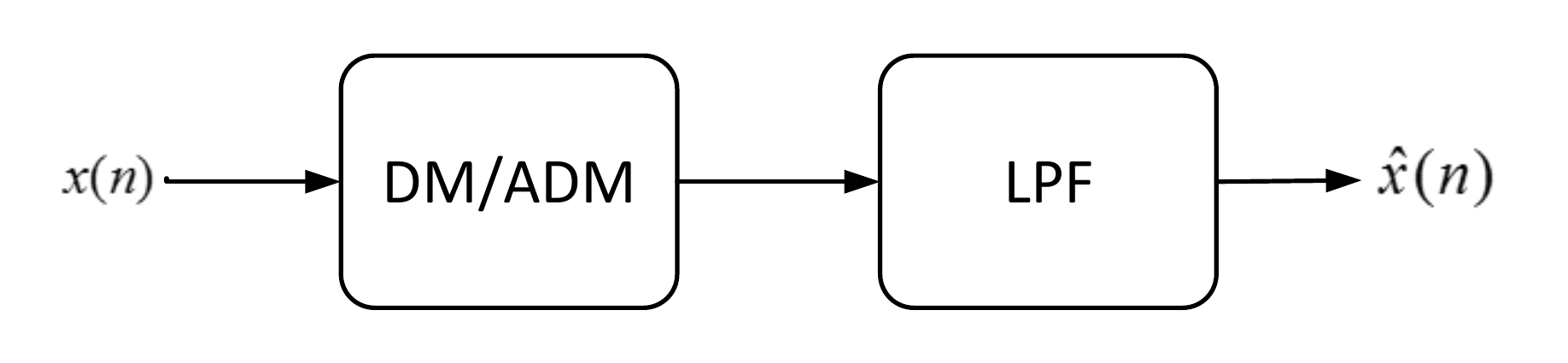}
		\label{fig:proposed}
	}
	\hspace{3mm}
	\centering
	\subfigure[The Proposed Method Using IMAT/IMATDM]
	{
		\includegraphics[width= 0.45\columnwidth]{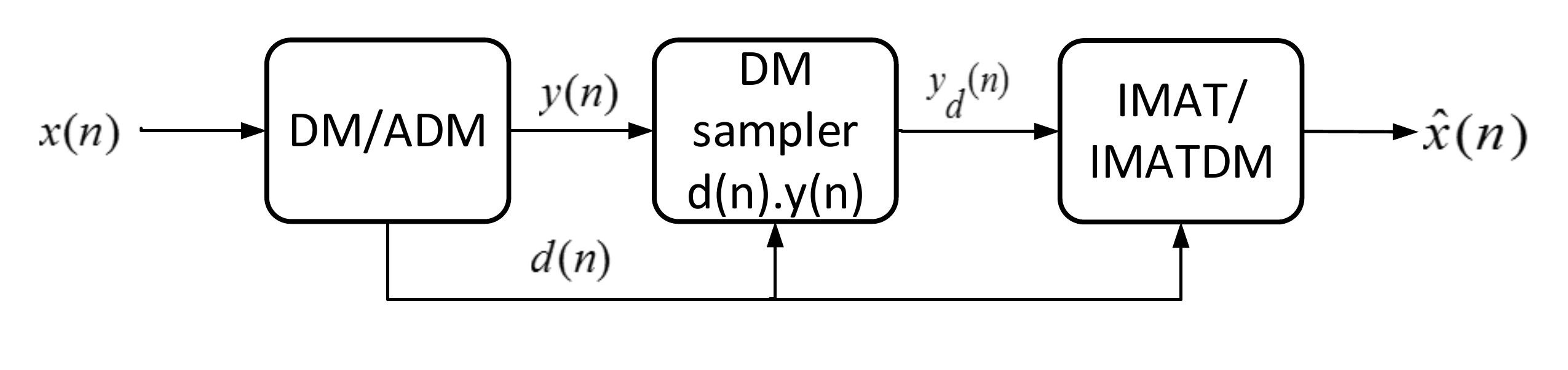}
		\label{fig:lpf}
	}
	\caption{The Conventional and the Proposed DM Reconstruction Techniques}
	\label{fig:comp}
	
\end{figure}

Fig. \ref{fig:comp} compares our proposed reconstruction method with the traditional low pass filtering techniques. As observed in this figure, the traditional reconstruction technique applies the low pass signal assumption. In order for this assumption to be true, the input signal is usually filtered by an antialiasing low pass filter prior to DM coding. This antialiasing filter omits the high frequency signal components and degrades the signal quality. However, the proposed sparsity promoting reconstruction algorithms of IMAT and IMATDM bypass this antialiasing filter distortion. Hence, IMAT/IMATDM algorithms improve the reconstruction quality by well preserving the high frequency signal components in the high quality input voice. This phenomenon is justified by simulations in section V.

\section{CONVERGENCE ANALYSIS}

In this section, we prove the convergence of IMAT and IMATDM algorithms for DM coding using our proposed missing sampling model. To this end, we need to prove the following theorem.  

\begin{theorem}{}\label{th:TH1}
The frequency transform of the delta modulation sampled signal $Y_{d}(m)=F\{y_{d}(n)\}$ is a random variable with mean and variance given by
\begin{align}\label{eq:TH1}
	E\{Y_{d}(m)\}&=pX(m)\nonumber
	\\E\{|Y_{d}(m)-pX(m)|^2\}&=(p-p^2)\epsilon_{x}+pN\frac{\Delta^2}{4}
\end{align}
in which $\epsilon_{x}=\sum_{n}{|x(n)|^{2}}$.
\end{theorem}

\begin{proof}	
According to the missing sampling model we assumed for DM coding, $y_{d}(n)$ is a random variable that takes one of the values $0$, $x(n)+\Delta/2$ and $x(n)-\Delta/2$ with probabilities $1-p$, $p/2$ and $p/2$ respectively. Hence, we get $E\{y_{d}(n)\}=px(n)$ and the first equation in (\ref{eq:TH1}) can be derived as (\ref{eq:PTH11})

\begin{align}\label{eq:PTH11}
E\{Y_{d}(m)\}&=E\{\sum_{n=0}^{N-1}y_{d}(n)\exp(-j\frac{2\pi}{N}nm)\}\nonumber
\\&=\sum_{n=0}^{N-1}E\{y_{d}(n)\}\exp(-j\frac{2\pi}{N}nm)\nonumber
\\&=\sum_{n=0}^{N-1}px(n)\exp(-j\frac{2\pi}{N}nm)\nonumber
\\&=pX(m)
\end{align}

This proves the expected value equation, for the variance equation we have

\begin{align}\label{eq:PTH12}
E\{|Y_{d}(m)-pX(m)|^2\}&=E\{(\sum_{n=0}^{N-1}(y_{d}(n)-px(n))e^{-j\frac{2\pi}{N}nm})\nonumber
\\&\times(\sum_{k=0}^{N-1}(y_{d}(k)-px(k))e^{j\frac{2\pi}{N}km})\}\nonumber
\\&=E\{\sum_{n}\sum_{k}(y_{d}(n)-px(n))\nonumber
\\&\times(y_{d}(k)-px(k))e^{j\frac{2\pi}{N}m(k-n)}\}\nonumber
\\&=\sum_{n=0}^{N-1}E\{(y_{d}(n)-px(n))^2\}\nonumber
\\&+\sum_{n}\sum_{k\neq n}E\{(y_{d}(n)-px(n))\nonumber
\\&\times(y_{d}(k)-px(k))\}e^{j\frac{2\pi}{N}m(k-n)}
\end{align}

Note that as $y_{d}(n)$ and $y_{d}(k)$ are assumed independent for $n\neq k$, the second term in (\ref{eq:PTH12}) equals zero. Again considering the simplified model we assumed for delta modulation sampling, the first term is calculated as (\ref{eq:PTH13})

\begin{align}\label{eq:PTH13}
\sum_{n=0}^{N-1}E\{(y_{d}(n)-px(n))^2\}&=\sum_{n=0}^{N-1} \{\frac{p}{2}\times((1-p)x(n)+\Delta/2)^2\nonumber
\\&+\frac{p}{2}\times((1-p)x(n)-\Delta/2)^2\nonumber
\\&+(1-p)\times p^2 x^2(n)\}\nonumber
\\&=\sum_{n=0}^{N-1} p(1-p)x^2(n)+p\frac{\Delta^2}{4}\nonumber
\\&=(p-p^2)\epsilon_{x}+pN\frac{\Delta^2}{4}
\end{align} 

And the proof is complete.
\end{proof}

\begin{theorem}{}\label{th:TH2}
Considering IMAT reconstruction formula given by (\ref{eq:IMAT}), $\lim_{k\to \infty} X_{k}(m)$ is an unbiased estimator of $X(m)$ for $0<\lambda<\frac{2}{p}$.
\end{theorem}

\begin{proof}
	To prove this theorem, we need to show that $E\{\lim_{k\to \infty} X_{k}(m)\}=X(m)$ or equivalently $\lim_{k\to \infty} E\{X_{k}(m)\}=X(m)$. To this end, we prove that the error sequence $e_k=X(m)-E\{X_{k}(m)\}$ forms a geometric progression with initial value $X(m)$ and common ratio $r=1-\lambda p$. Equivalently, IMAT reconstruction technique converges linearly (of order 1) to the original signal in the mean iff $|r|=|1-\lambda p|<1$ or $0<\lambda<2/p$. 
	
	Starting the algorithm from zero initial value, we have $X_{0}(m)=0$ and hence $e_0=X(m)-E\{X_{0}(m)\}=X(m)$. the statement is trivial for $k=0$. According to (\ref{eq:IMAT}) for $k=1$, we have $x_{1}(n)=\lambda y_{d}(n)$. Also applying theorem \ref{th:TH1} we get $E\{X_{1}(m)\}=\lambda pX(m)$ and hence $e_1=(1-\lambda p)X(m)$. Now, we rewrite the iterative IMAT formula (\ref{eq:IMAT}) in transform domain as 
	
	\begin{align}\label{eq:rabete2}
	X_{k+1}(m)=\lambda Y_{d}(m)+F\{(1-\lambda d(n))T(x_{k})(n)\}(m)
	\end{align}
	
	Now, taking $E\{.\}$ from both sides of (\ref{eq:rabete2}) and applying (\ref{eq:TH1}) we get
	
	\begin{align}\label{eq:rabete3}
	E\{X_{k+1}(m)\}&=\lambda pX(m)+F\{E\{(1-\lambda d(n))T(x_{k})(n)\}\}(m)
	\end{align}
	
	Utilizing (\ref{th:TH1}) in (\ref{eq:rabete3}) we get
	
	\begin{align}\label{eq:rabete4}
		E\{X_{k+1}(m)\}=\lambda pX(m)+(1-\lambda p)F\{E\{T(x_{k})(n)\}\}(m)
	\end{align}
	
	Utilizing (\ref{eq:rabete4}) we get
	
	\begin{align}\label{eq:rabete5}
	X(m)-E(X_{k+1}(m))&=(1-\lambda p)X(m)
	\\&-(1-\lambda p)E(F\{T(x_{k})(n)\}(m))\nonumber	
	\end{align}
	
	Now if $|X_{k}(m)| \geq TH(k)$, the thresholding operator can be omitted from the right side of (\ref{eq:rabete4}) which yields (\ref{eq:rabete6}) and the required statement results.
	
	\begin{align} \label{eq:rabete6}
		e_{k+1}&=X(m)-E(X_{k+1}(m))\nonumber
		\\&=(1-\lambda p)X(m)-(1-\lambda p)E(X_k(m))\nonumber
		\\&=(1-\lambda p) e_k
	\end{align}
	
	On the other hand, if $X_{k}(m)$ is not picked by the threshold, we get $E(F\{T(x_{k})(n)\}(m))=0$ and according to (\ref{eq:rabete4}) we get $e_{k+1}=(1-\lambda p)X(m)$. Hence, once a signal component is picked by the threshold in a specific iteration, its corresponding error sequence converges linearly to zero. As the threshold is strictly decreasing, all signal components will be gradually picked by the threshold and the proof is complete.
\end{proof}

 Theorem \ref{th:TH2} proves that $X_{k+1}(m)$ is a random variable with its expected value approaching $X(m)$ as $k \to \infty$. But in order to prove perfect reconstruction/convergence of the IMAT algorithm, we also need to show that the variance of this unbiased estimator approaches zero as $k \to \infty$. Theorem 2 explains this variance fluctuation issue as k approaches infinity. Before providing the formal statement for theorem \ref{th:TH3} lets define the sparse signal support to be the set of all its nonzero frequency components as $Supp=\{m|X(m) \neq 0\}$. 
 
 \begin{theorem}{}\label{th:TH3}
 	Assume that the threshold picks $X_k(m_j)$ in the k'th iteration of the IMAT algorithm, this increases the spectrum variance if $m_j \in Supp$ and increases the variance for $m_j \notin Supp$. 
 \end{theorem}
 
 \begin{proof}
 	Let's decompose $T(x_k(n))$ as 
 	
 	\begin{align}\label{eq:rabete7}
 	T(x_k(n))=p_{k+1}(n)+s_{k+1}(n)
 	\end{align}
 	
 	In which $p_{k+1}(n)$ only consists of signal support $Supp$ components and $s_{k+1}(n)$ includes other components picked by the threshold. In other words, $p_{k+1}(n)$ is the portion reconstructed by the algorithm up to the k'th iteration and $s_{k+1}(n)$ is the non-support portion picked mistakenly by the threshold due to non-zero spectrum variance. Now, let's also decompose $x(n)$ as the sum of its reconstructed portion $p_{k+1}(n)$ and a residual $r_{k+1}(n)$ as
 	\begin{align}\label{eq:rabete8}
 	x(n)=p_{k+1}(n)+r_{k+1}(n)
 	\end{align}
 	
 	Now, note that according to the missing sampling model for DM coding we have $y_{d}(n)=d(n)(x(n)+q(n))$. Substituting (\ref{eq:rabete7}) and (\ref{eq:rabete8}) in (\ref{eq:IMAT})gives
 	
 	\begin{align} \label{eq:rabete9}
 	x_{k+1}(n)&=\lambda y_d(n)+(1-\lambda d(n))T(x_k(n))\nonumber
 	\\&=\lambda d(n)(x(n)+q(n))+(1-\lambda d(n))T(x_k(n))\nonumber
 	\\&=\lambda d(n)(p_{k+1}(n)+r_{k+1}(n))+\lambda d(n)q(n)\nonumber
 	\\&+(1-\lambda d(n))(p_{k+1}(n)+s_{k+1}(n))\nonumber
 	\\&=p_{k+1}(n)+s_{k+1}(n)+\lambda d(n)q(n)\nonumber
 	\\&+\lambda d(n)r_{k+1}(n)-\lambda d(n)s_{k+1}(n)
 	\end{align}
 	
 	Now, note that the first two terms in (\ref{eq:rabete9}) are not multiplied by the sampling mask $d(n)$ and hence their frequency transform are deterministic and do not contribute to the spectrum variance. According to theorem \ref{th:TH1}, the contribution of $\lambda d(n)r_{k+1}(n)$ and $\lambda d(n)s_{k+1}(n)$ terms to the spectrum variance equals $\lambda ^2 (p-p^2) \epsilon _{r_{k+1}}$ and $\lambda ^2 (p-p^2) \epsilon _{s_{k+1}}$ respectively. Similarly, according to theorem \ref{th:TH1}, the contribution of the coding error to the spectrum variance equals $\lambda ^2 pN \frac{\Delta ^2}{4}$. This yields that
 	\begin{align} \label{eq:rabete10}
 	\sigma ^2_{k+1}&=E\{(X_{k+1}(m)-E\{X_{k+1}(m)\})^2\}\nonumber
 	\\&=\lambda ^2 (p-p^2) \epsilon _{r_{k+1}}\nonumber
 	\\&+\lambda ^2 (p-p^2) \epsilon _{s_{k+1}}\nonumber
 	\\&+\lambda ^2 pN \frac{\Delta ^2}{4}
 	\end{align}
 	
 	In (\ref{eq:rabete10}), $\epsilon_{r_{k+1}}$ and $\epsilon_{s_{k+1}}$ denote the energies of the residual and mistakenly picked frequency components respectively. Now, each mistakenly picked component $m_j \notin Supp$ increases $\epsilon_{s_{k+1}}$ and consequently the spectrum variance. Similarly, for a correctly picked signal component $m_j \in Supp$, $\epsilon_{r_{k+1}}$ and consequently the spectrum variance $\sigma ^2_{k+1}$ is decreased. The above discussion completes the proof. 
\end{proof}

\begin{remark} \label{rm:RM1}
	As stated previously, due to the non-zero spectrum variance, $X_k(m_j)$ is non-zero for $m_j \notin Supp$. Hence, the threshold parameters must be adjusted such that the threshold value always keeps above the standard deviation of the spectrum (e.g. $Th(k) \geq \gamma \sigma _k$, $\gamma >1$) to prevent the algorithm from picking incorrect frequency components. In this case, $\epsilon _{s_{k+1}}=0$ and the spectrum variance is decreasing in each iteration $\sigma^2_{k+1} \leq \sigma^2_{k}$.
\end{remark}

\begin{corollary}
Considering theorem \ref{th:TH2}, we conclude that IMAT estimation bias approaches zero as $k$ approaches infinity. On the other hand, the variance of IMAT estimation is decreasing provided that the condition in Remark \ref{rm:RM1} always holds. Now considering the fact
that the Mean Square Error (MSE) of the estimator is given by (\ref{eq:rabete11}) 
\begin{align} \label{eq:rabete11}
MSE_k&=E\{(X_k(m)-X(m))^2\}\nonumber
\\&=(E\{X_k(m)\}-X(m))^2+\sigma^2_k
\end{align}
Now as both terms in (\ref{eq:rabete11}) are decreasing, we conclude that $MSE_k$
is also decreasing. In other words, the cost function for
IMAT algorithm is the Mean Square Error of the estimated
signal in transform domain. As this cost function is convex and
decreasing, it will converge to the global minimum.  
\end{corollary}

\begin{remark}
	If the threshold value is decreased too fast, there
	exists the risk of picking signal components that are not in
	signal support. This will increase the spectrum variance or
	equivalently degrade the final quality of the reconstructed
	signal. On the other hand, slower decrease of the threshold value, slows down the algorithm’s convergence rate. Although the threshold doesn’t necessarily need to decrease exponentially, simulation results conveyed that the exponential parameters $(\alpha,\beta)$ can always be optimized for acceptable performance.
\end{remark}
  
\section{SIMULATION RESULTS}
In this section, we first investigate the performance of the proposed missing sampling approach to DM coding in subsection A. In subsection B, we show superior performance of the proposed IMAT and IMATDM algorithms for DM/ADM reconstruction in comparison with both the classic low pass and some other sparsity promoting  reconstruction algorithms.

 It should be noted that we compare the performance of different sampling/reconstruction schemes in terms of the average reconstruction SNR and Success Rates (SR) achieved on a statistically sufficient number of voice frames. Each frame is 20ms long and is extracted from a high quality voice recorded at 48kHz. Each frame is 20ms long and is extracted from a database of high quality English
 voice (55\% male, 45\% female) sampled at 48kHz. Hence, the sampling frequency of the input voice is set
 at 48kHz to ensure that no frequency component in human audible range (20Hz to 20 kHz) is missed.
 
 To gain an insight on the subjective quality achieved by different techniques, we also report PESQ (Perceptual Evaluation of Speech Quality) \cite{PESQ} scores achieved on the whole 10s voice recording. The PESQ score ranges from -0.5 to 4.5 with 4.5 standing for the best subjective quality.
 
 The SNR values reported as the quality measure in this section are calculated as (\ref{eq:rabete12})
 
 \begin{align} \label{eq:rabete12}
 SNR= \frac{\sum_{i=0}^{N-1} x(i)^2}{\sum_{i=0}^{N-1}\left( x_k(i)-x(i) \right)^2}
 \end{align}
 
 The reported SR values are also defined as the percentage of the voice frames reconstructed with final SNRs above a preset threshold value.
 
 Finally, note that a central contribution in this paper is to replace the traditional low pass signal assumption
 (e.g. 3.3 kHz as the bandwidth for voice) by the sparse signal assumption using the family of IMAT based
 reconstruction techniques. Hence, unlike the common assumption of 3.3 kHz for voice, we avoid the antialiasing
 filter on 3.3 kHz to be able to capture frequency components beyond 3.3 kHz. Consequently, the
 input sampling frequency must be so high (e.g. 48 kHz) as to make sure that all components in the human
 audible range [20-20k] Hz are accommodated. Although these high frequency components may seem
 negligible, but we observe that utilizing the sparse assumption, we achieve noticeable SNR and PESQ
 improvements in comparison with the traditional low pass based reconstruction techniques.

\subsection{Delta Modulation vs. Missing Sampling}

As mentioned in the previous sections, we utilize the missing sampling model for DM coding both in reconstruction and analytical studies. Off course there exists some modeling mismatch in our proposed approach. To study the effects of this modeling mismatch, we compare these two sampling methods in terms of the average reconstruction SNR values achieved by them. It should be noted that in this subsection, we compare the reconstruction performance of the two sampling patterns/masks of missing and DM sampling. Hence we need to discard the quantization/coding error. The two sampling patterns are compared at equal sampling rates.

For DM sampling, as mentioned before, the sampling pattern ($d(n)$) is extracted according to (\ref{eq:mask}). But in order to discard the coding errors in reconstruction, $y_d(n)$ is replaced by $x_d(n)=d(n).x(n)$ in (\ref{eq:IMAT}). For the missing sampling approach, the mask elements $m(n)$'s are generated by iid random variables of $Bernoulli(p)$ distribution in which $p$ is the same sampling rate achieved by DM coding. Again, $y_d(n)$ is replaced by $x_m(n)=m(n).x(n)$ for reconstruction without quantization error.  

The comparison is performed at different delta values of $\Delta=$ 0.001, 0.005, 0.01 and 0.03 that yield the sampling rates of $p=$ 0.46, 0.70, 0.80 and 0.88, respectively. The simulation results has been shown in fig \ref{fig:wonoise}. According to this figure, as the delta parameter approaches 0.03, different reconstruction methods achieve almost equal results in terms of SNR for both random sampling and our proposed DM based missing sampling.

\begin{figure}[!ht]
	\centering
	\includegraphics[width=0.5\linewidth]{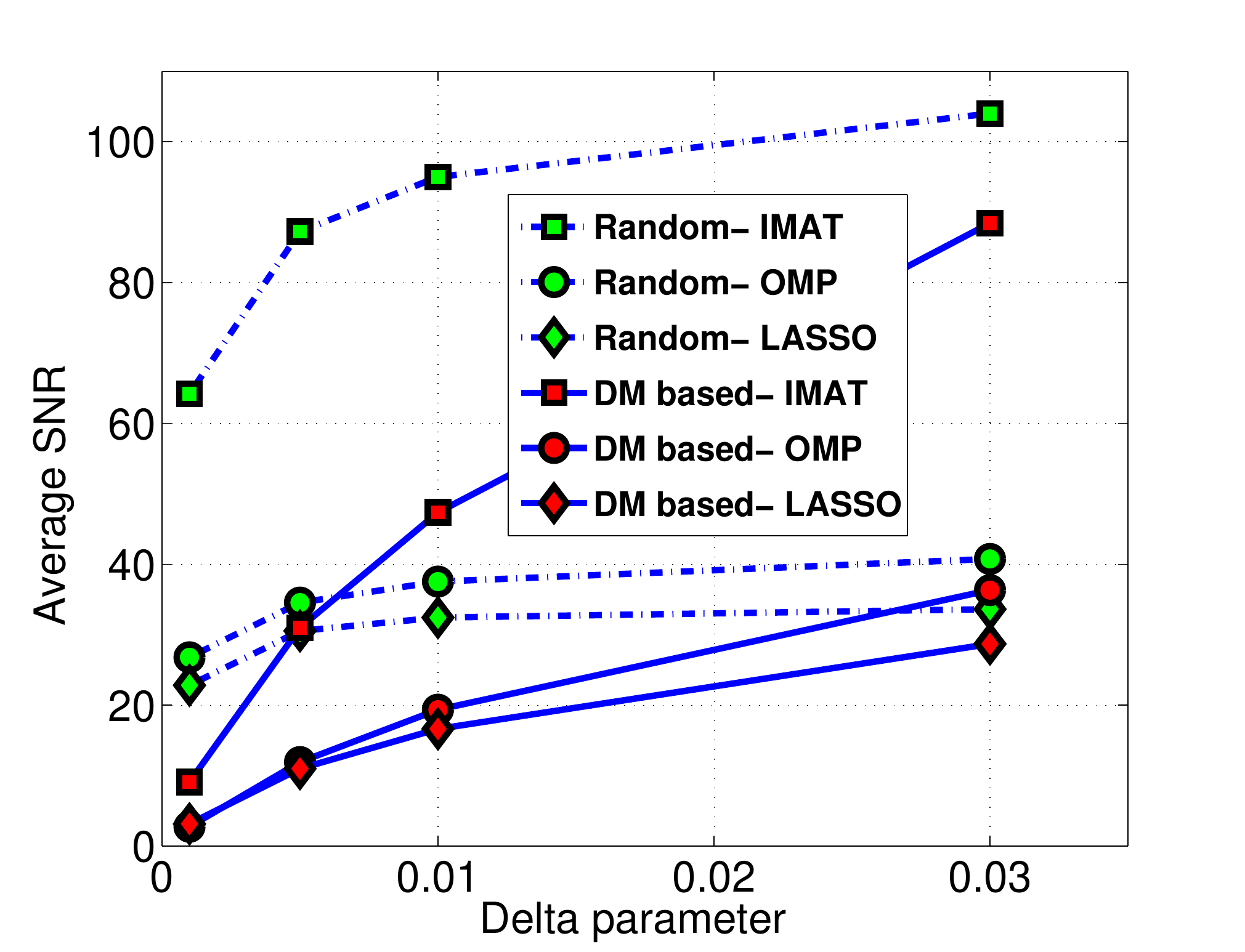}
	\caption{Comparsion between the Ideal Missing Sampling and DM Sampling in Terms of Reconstruction SNR (dB)}
	\label{fig:wonoise}
\end{figure}

It can be concluded that DM is a suitable sampling method and its performance approaches the ideal case of random sampling as the sampling rate increases. 

\subsection{Performance Comparisons between Different Reconstruction Methods}

In this subsection we compare the performance of the proposed IMAT and IMATDM algorithms with both the classic low pass reconstruction technique and other sparsity promoting algorithms of  OMP and LASSO. To this end, we apply DM coding at different $\Delta$ values leading to different sampling rates. It should be noted that unlike the previous subsection the coding error is taken into account in the simulations of this subsection to facilitate a more meaningful comparison between different reconstruction techniques.

Figure \ref{fig:deltadmwnoise}  compares the average SNR (dB) values achieved by these algorithms for DM sampling with different $\Delta$ values. The comparison is performed at delta values of $\Delta = $ 0.001, 0.005, 0.01, 0.02, and 0.03 that lead to sampling rates of $ p =$ 0.46, 0.70, 0.80, 0.86, and 0.88 respectively. 

\begin{figure}[!ht]
	\centering
	\includegraphics[width=0.45\linewidth]{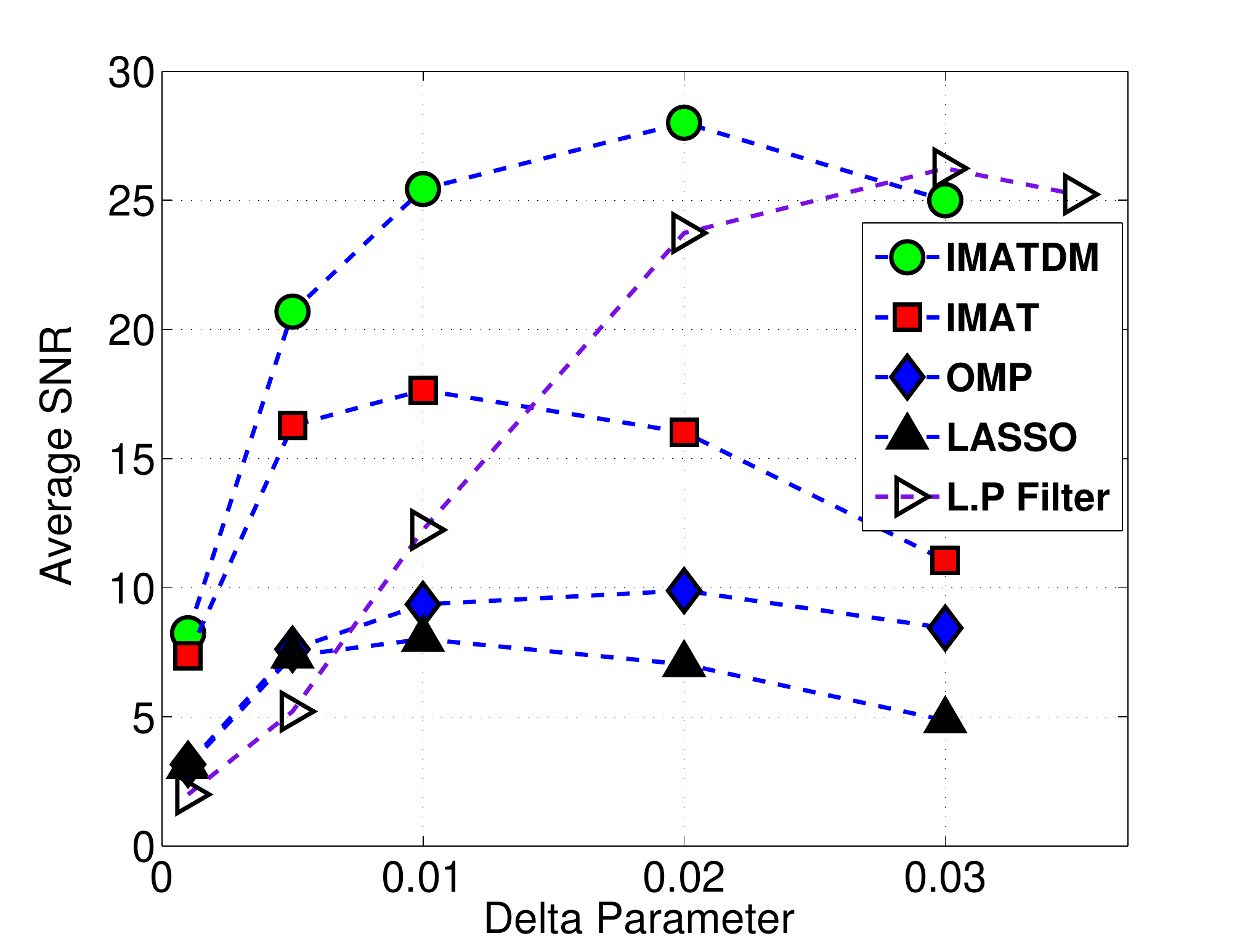}
	\caption{Comparison between Different Reconstruction Methods for DM Sampling in Terms of Reconstruction SNR (dB)}
	\label{fig:deltadmwnoise}
\end{figure}

The corresponding Success Rate (SR) values defined as the ratio of successfully reconstructed voice frames in $(\%)$ are also reported in table \ref{jadvala} . In this table, the SNR threshold considered as successful reconstruction is set at 15 dBs.
 \begin{table}[ht]
 	
 	\centering
 	\caption{Success rates$(\%)$ achieved by different reconstruction methods for DM (SNR values exceeding 15 $dB$ are considered as successful reconstruction) }
 	\label{jadvala}
 	
 	\begin{tabular}{c}
 		\includegraphics[width=0.6\linewidth]{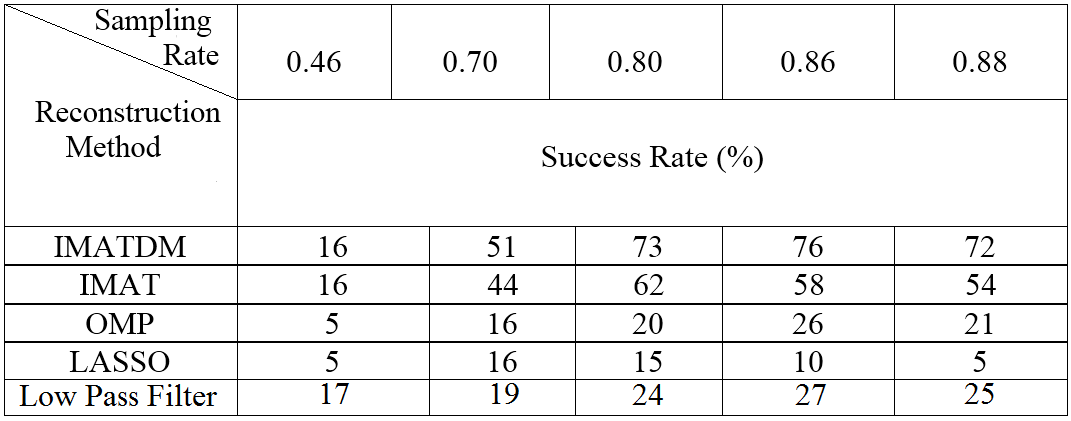}
 	\end{tabular}
 	
 \end{table}

 It can be concluded that the proposed reconstruction algorithms (IMAT and IMATDM) clearly outperform both sparsity promoting algorithms (OMP and LASSO) and low pass technique regarding SNR and SR.

 Not that as an increase in $\Delta$ leads to an increase in both sampling rate and the coding error, there is an optimum value for $\Delta$ in each scenario. This phenomenon is also observed in fig. \ref{fig:deltadmwnoise}.

 Finally, the PESQ scores achieved by IMATDM, IMAT, LASSO, OMP, and the classical Low pass techniques for DM reconstruction at p= 0.8 are 2.33, 2.08, 1.44, 1.1, and 1.49, respectively.
 
 In order to achieve more comprehensive results, we also investigate the performance of the proposed reconstruction algorithms on the missing sampling based on ADM. It should be noted that as the $\Delta$ parameter is changed adaptively according to the input signal for ADM, the sampling rate for ADM is less dependent on the initial value of $\Delta$ and is almost $p=$ 0.5 . The reconstruction SNR values achieved for ADM was 39.7, 31.9, 12.9, 14.21, and 39.5 dBs for IMATDM, IMAT, OMP, LASSO, and the classical Low Pass techniques, respectively. Similarly, the achieved SR values are 88.5, 52.1, 19.5, 9.5 and 40.8 for IMATDM, IMAT, OMP, LASSO, and the classical Low Pass techniques, respectively. Note that the SNR threshold considered as successful reconstruction is set at 20 dBs for SR calculations. 
 \\
 
 Finally, the PESQ scores achieved by IMATDM, IMAT, OMP, LASSO, and the classical Low pass technique are 3.1, 2.67 ,2.08 ,2.12 , and 2.9 in this scenario. 
 \\
 
 It is observed that, sampling by ADM and reconstructing by IMATDM achieves superior performance in comparison with other sampling/reconstruction pairs. It is also observed that the proposed IMATDM algorithm outperforms the basic IMAT by an average amount of 7.83 dBs.

\section{CONCLUSION}
In this paper, we proposed the family of iterative methods with adaptive thresholding (IMAT) for sparsity promoting reconstruction of Delta Modulated (DM)
and Adaptive Delta Modulated (ADM) voice signals. By adding a prior partial smoothing block to IMAT, we proposed IMATDM algorithm for improved reconstruction quality for DM sampling. Utilizing our proposed missing sampling approach to delta modulation, we proved the convergence of IMAT and IMATDM algorithms analytically. The proposed algorithms also proved promoting reconstruction techniques of OMP and LASSO in terms of reconstruction SNR, Success Rate and PESQ scores by simulations. 

\bibliographystyle{IEEEtran}
\bibliography{Citations}

\begin{thebibliography}{10}
\providecommand{\url}[1]{#1}
\csname url@samestyle\endcsname
\providecommand{\newblock}{\relax}
\providecommand{\bibinfo}[2]{#2}
\providecommand{\BIBentrySTDinterwordspacing}{\spaceskip=0pt\relax}
\providecommand{\BIBentryALTinterwordstretchfactor}{4}
\providecommand{\BIBentryALTinterwordspacing}{\spaceskip=\fontdimen2\font plus
\BIBentryALTinterwordstretchfactor\fontdimen3\font minus
  \fontdimen4\font\relax}
\providecommand{\BIBforeignlanguage}[2]{{%
\expandafter\ifx\csname l@#1\endcsname\relax
\typeout{** WARNING: IEEEtran.bst: No hyphenation pattern has been}%
\typeout{** loaded for the language `#1'. Using the pattern for}%
\typeout{** the default language instead.}%
\else
\language=\csname l@#1\endcsname
\fi
#2}}
\providecommand{\BIBdecl}{\relax}
\BIBdecl

\bibitem{shannon}
C.~Shannon, ``A mathematical theory of communication,'' \emph{Bell Syst Tech.
  J.}, vol.~27, 1948.

\bibitem{jerri}
A.~Jerri, ``Shannon sampling theorem-its various extensions and applications: A
  tutorial review,'' \emph{Proc. IEEE}, vol.~65, 1977.

\bibitem{MYLC}
M.~B. Mashhadi, N.~Salarieh, E.~S. Farahani, and F.~Marvasti, ``Level crossing
  speech sampling and its sparsity promoting reconstruction using an iterative
  method with adaptive thresholding,'' \emph{IET Signal Processing}, vol.~11,
  no.~6, pp. 721--726, 2017.

\bibitem{MYPATENT}
F.~Marvasti and M.~B. Mashhadi, ``Wideband analog to digital conversion by
  random or level crossing sampling,'' \emph{US Patent No. US9729160B1}, 2017.

\bibitem{nusampling}
Y.~Hu and et. al., ``Reconstruction of uniformly sampled signals from
  non-uniform short samples in fractional fourier domain,'' \emph{IET Signal
  Processing}, November 2015.

\bibitem{Drmarvasti2001nonuniform}
F.~Marvasti, in \emph{Nonuniform Sampling: Theory and Practice}.\hskip 1em plus
  0.5em minus 0.4em\relax Springer US, 2001.

\bibitem{DMMOD}
J.~Flood and M.~Hawksford, ``Exact model for delta-modulation processes,''
  \emph{Proceedings of the Institution of Electrical Engineers, IET}, vol. 118,
  September 1971.

\bibitem{ADMIET}
P.~Wing, ``Adaptive delta modulation,'' \emph{Electronics Letters}, vol.~5, pp.
  191--192(1), May 1969.

\bibitem{Zhu_paper}
Y.~Zhu and S.~Leung, ``A nonuniform sampling delta modulation technique,'' in
  \emph{TENCON '94. IEEE Region 10's Ninth Annual International Conference.
  Theme: Frontiers of Computer Technology. Proceedings of 1994}, 1994, pp.
  708--712 vol.2.

\bibitem{ADM}
J.~Abate, ``Linear and adaptive delta modulation,'' \emph{Proceedings of the
  IEEE}, vol.~55, no.~3, pp. 298--308, 1967.

\bibitem{app}
S.-C. Kim, C.-P. Fan, J.-F. Yang, B.-D. Liu, and D.~Yang, ``Block delta
  modulation for compression of surveillance video signals,'' \emph{IEEE
  Transactions on Consumer Electronics}, vol.~40, no.~3, pp. 458--465, 1994.

\bibitem{Reza}
R.~{Kazemi}, M.~{Boloursaz}, S.~M. {Etemadi}, and F.~{Behnia}, ``Capacity
  bounds and detection schemes for data over voice,'' \emph{IEEE Transactions
  on Vehicular Technology}, vol.~65, no.~11, pp. 8964--8977, Nov 2016.

\bibitem{Soheil}
S.~{Salehi}, M.~B. {Mashhadi}, A.~{Zaeemzadeh}, N.~{Rahnavard}, and R.~F.
  {DeMara}, ``Energy-aware adaptive rate and resolution sampling of spectrally
  sparse signals leveraging vcma-mtj devices,'' \emph{IEEE Journal on Emerging
  and Selected Topics in Circuits and Systems}, vol.~8, no.~4, pp. 679--692,
  Dec 2018.

\bibitem{noise}
F.~Sakane and R.~Steele, ``Estimation of the signal-to-noise ratio in high
  information and constant factor delta modulation systems,'' \emph{IEEE
  Transactions on Communications}, vol.~25, no.~12, pp. 1441--1448, 1977.

\bibitem{qnoise}
H.~S. Lee and C.~Un, ``Quantization noise in adaptive delta modulation
  systems,'' \emph{IEEE Transactions on Communications}, vol.~28, no.~10, pp.
  1794--1802, 1980.

\bibitem{grannoise}
J.~David and Goodman, ``Delta modulation granular quantizing noise,''
  \emph{Bell System Technical Journal, The , vol.48, no.5}, pp. 1197--1218,
  May-June 1969.

\bibitem{lprecons}
A.~A. Lazar and L.~Toth, ``Perfect recovery and sensitivity analysis of time
  encoded bandlimited signals,'' \emph{IEEE Transactions on Circuits and
  Systems I: Regular Papers, vol. 51, no. 10}, pp. 2060--2073, Oct. 2004.

\bibitem{Dramini}
F.~Marvasti, A.~Amini, and F.~Haddadi, ``A unified approach to sparse signal
  processsing,'' \emph{EURASIP Journal on Advances in Signal Prcessing}, 2012.

\bibitem{Omp}
A.~Tropp and A.~C. Gilbert, ``Signal recovery from partial information via
  orthogonal matching pursuit,'' \emph{IEEE Transaction on Information Theory,
  vol. 53, No. 12}, 2007.

\bibitem{lasso}
R.~Tibshirani, ``Regression shrinkage and selection via the lasso,''
  \emph{Journal of the Royal Statistical Society}, 1996.

\bibitem{IMAT}
F.~Marvasti, M.~Azghani, and et~al, ``Sparse signal processing using iterative
  method with adaptive thresholding(imat),'' \emph{19th International
  Conference on Telecommunication(ICT)}, 2013.

\bibitem{Prosalentis2}
E.~A. Prosalentis and G.~S. Tombras, ``A 2-bit adaptive delta modulation system
  with improved performance,'' \emph{EURASIP Journal on Advances in Signal
  Processing}, November 2006.

\bibitem{Ashi1}
A.~Esmaeili, E.~A. Kangarshahi, and F.~Marvasti, ``Iterative null space
  projection method with adaptive thresholding in sparse signal recovery,''
  \emph{IET Signal Processing}, vol.~12, no.~5, pp. 605--612, 2018.

\bibitem{PESQ}
Y.~Hu and P.~Loizou, ``Evaluation of objective quality measures for speech
  enhancement,'' \emph{IEEE Transactions on Audio, Speech, and Language
  Processing}, vol.~16, no.~1, pp. 229--238, 2008.

\end{thebibliography}

\end{document}